\documentclass[final,3p,times,onecolumn]{elsarticle}

\usepackage{amssymb}
\usepackage{lipsum}
\usepackage{mathtools}
\usepackage{textcomp,booktabs}
\usepackage{enumerate}
\usepackage{enumitem}
\usepackage{url}
\usepackage[ruled,vlined,linesnumbered]{algorithm2e}
\usepackage{xcolor}
\usepackage{multirow}

\usepackage{mathrsfs}
\usepackage{amsthm}
\theoremstyle{remark}

\newtheorem{theorem}{Theorem}

\newtheorem{corollary}{Corollary}
\newtheorem{assumption}{Assumption}
\newtheorem{prop}{Proposition}

\SetKwInput{KwInput}{Input}                
\SetKwInput{KwOutput}{Output}

\begin{document}

\begin{frontmatter}



\title{Connection-Aware P2P Trading: Simultaneous Trading and Peer Selection\tnoteref{label1}}


\author[first]{Cheng Feng}
\author[first]{Kedi Zheng}
\author[first]{Lanqing Shan}
\author[2]{Hani Alers}
\author[first]{Qixin Chen}
\author[2]{Lampros Stergioulas}
\author[first]{Hongye Guo\corref{cor1}}

\affiliation[first]{organization={Department of Electrical Engineering, Tsinghua University, Beijing 100084, China}}
\affiliation[2]{organization={The Hague University of Applied Sciences, 2521 EN Den Haag, Netherlands}}

\cortext[cor1]{Corresponding Author: Hongye Guo (Email: hyguo@mail.tsinghua.edu.cn)}
\tnotetext[label1]{This work is supported in part by the DATALESs project jointly financed by National Natural Science Foundation of China (NSFC) and Netherlands Organization for Scientific Research (NWO), and in part by the Special Funds of National Natural Science Foundation of China. The NSFC grant numbers are 52161135201 and 72342007. The NWO project number is 482.20.602.}

\begin{abstract}
Peer-to-peer (P2P) trading is seen as a viable solution to handle the growing number of distributed energy resources in distribution networks. However, when dealing with large-scale consumers, there are several challenges that must be addressed. One of these challenges is limited communication capabilities. Additionally, prosumers may have specific preferences when it comes to trading. Both can result in serious asynchrony in peer-to-peer trading, potentially impacting the effectiveness of negotiations and hindering convergence before the market closes. This paper introduces a connection-aware P2P trading algorithm designed for extensive prosumer trading. The algorithm facilitates asynchronous trading while respecting prosumer's autonomy in trading peer selection, an often overlooked aspect in traditional models. In addition, to optimize the use of limited connection opportunities, a smart trading peer connection selection strategy is developed to guide consumers to communicate strategically to accelerate convergence. A theoretical convergence guarantee is provided for the connection-aware P2P trading algorithm, which further details how smart selection strategies enhance convergence efficiency. Numerical studies are carried out to validate the effectiveness of the connection-aware algorithm and the performance of smart selection strategies in reducing the overall convergence time.
\end{abstract}




\end{frontmatter}




\section{Introduction}
\label{introduction}

With the blooming of distributed energy resources (DERs) at the distribution grid, traditional electricity consumers are becoming proactive `prosumers', actively participating in peer-to-peer (P2P) trading with each other~\cite{TusharSaha-1}.  P2P trading can contain local power fluctuations, alleviate the burden of balancing, and provide flexibility services to the main grid~\cite{ShengWang-2}.

P2P trading typically involves several iterations for prosumers to finalize trading proposal agreements in a decentralized manner~\cite{FengLiang-3}. As the scale of P2P trading expands, a single prosumer might interact with numerous potential trading peers, significantly increasing the complexity of information exchange~\cite{MoretBaroche-6}. Communcation constraints limit simultaneous negotiations~\cite{UmerHuang-7}, while individual preferences lead to selective peer engagement~\cite{SousaSoares-5}. This requires asynchronous
P2P trading approaches. Current trading mechanisms, which are predominantly synchronous or overlook the self-selection process during trading, risk reducing welfare and impeding effective trading agreements~\cite{Nguyen-8}. 

Recent research highlights the importance of P2P trading in future distribution networks. Part of this research focused on a more detailed modeling of prosumers, considering their strategic behaviors~\cite{K.S.-13}, bounded rationality~\cite{Y.Q.-10}, and used learning techniques to help prosumers to perform trades~\cite{qiu2021scalable}. Others incorporated a more detailed network model into trading problems, employing sensitivity factors~\cite{GuerreroChapman-11} and DistFlow~\cite{M.J.-14}. These models are valuable references for P2P trading analysis. Meanwhile, the implementation and development of solutions to trading models is also of significant importance. Various information exchange structures are designed to facilitate P2P trading, including leader-follower~\cite{M.M.-17}, bilateral~\cite{J.Y.-15}, multi-layer~\cite{CuiWang-18}, and grand coalition schemes~\cite{LiYe-19}. This paper adopts a decentralized structure for its privacy protection and independence from a coordination entity~\cite{L.A.-20}. Notably, most of the trading algorithms mentioned above, regardless of the information structure, neglect connection limits and require strict synchronization between trading peers or between the prosumer and the coordinator.

\color{black}
Addressing connection issues in P2P trading is critical. One type of related research focused on communication connection formation problems in prosumer coalitions. Some studies \cite{TusharSaha-21} illustrated that the trading connections among prosumers are flexible and can be adjusted according to prosumers' preferences. One study proposes a graph coalition formation game \cite{TusharSaha-21}, while another suggests peer matching \cite{KhorasanyPaudel-1351}. These algorithms are mainly designed for P2P trading in a small community. They can be hard to be extended to handle large community because of their iterative complexity.  

Meanwhile, the connection issue is also a problem for non-coalition P2P trading. A few studies have investigated the non-coalition P2P trading connection problem. Ref.\cite{ChenYang-27} offered an overview of information-related issues in P2P trading, and Ref.\cite{DongBaroche-28} conducted extensive numerical experiments on this topic. Ref.\cite{UllahPark-25,Z.P.-29} pioneered the use of asynchronous alternating direction method of multipliers (ADMM) to counter random communication issues in P2P trading. Nevertheless, Ref.\cite{UllahPark-25} lacked a theoretical convergence guarantee for the constrained P2P trading problem. Ref.\cite{Z.P.-29} introduced complex consensus processes during trading, exacerbating connection issues. Some works \cite{LiuXu-32,LiLi-33} adopted an even-triggered version of the ADMM algorithm to address connection issues, named Communication-Sensored Consensus ADMM (COCA)~\cite{LiuXu-30}. In COCA, the connection peers are activated based on a diminishing threshold. Like all event-trigger algorithms, determining this threshold in real-world scenarios is challenging. Besides, the diminishing threshold should converge to zero, potentially activating almost all links over time.  Ref.\cite{umer2021novel} used the node coloring algorithm to design a communication-efficient P2P trading mechanism. However, these works all assumed random connection problems. In contrast, the prosumers are not just nodes in communication networks, they have inherently different utilities and preferences. Existing algorithms ignored prosumers' willingness and overlooked their self-selection rights. A complete of the existing research about communication issues in P2P trading is shown in Table \ref{tab:research}. 
\color{black}

\begin{table}[h]
\color{black}
\caption{A comparison of P2P trading solutions against communication issues across different algorithms.}
\label{tab:research}
\resizebox{\textwidth}{!}{%
\begin{tabular}{ccccccc}
\hline
\multicolumn{1}{l}{} &
  \multicolumn{1}{l}{} &
  \multicolumn{1}{c}{\begin{tabular}[c]{@{}c@{}}Handle\\ Asynchrony\end{tabular}} &
  \multicolumn{1}{c}{\begin{tabular}[c]{@{}c@{}}Handle\\ Self-selection\end{tabular}} &
  \multicolumn{1}{c}{\begin{tabular}[c]{@{}c@{}}Handle\\ Large Community\end{tabular}} &
  \multicolumn{1}{c}{\begin{tabular}[c]{@{}c@{}}Handle\\ Long-time Negotiation\end{tabular}} &
  \multicolumn{1}{c}{\begin{tabular}[c]{@{}c@{}}Theoretical\\ Guarantee\end{tabular}} \\ \hline
\multirow{2}{*}{\begin{tabular}[c]{@{}c@{}}Coalition\\ Game\end{tabular}}         & Graph Formation\cite{TusharSaha-21} & \checkmark  & \checkmark & $\times$  & $\times$ & $\times$\\
& Peer Matching\cite{KhorasanyPaudel-1351}     & \checkmark &\checkmark & $\times$ & $\times$ &\checkmark\\ \hline
\multirow{5}{*}{\begin{tabular}[c]{@{}c@{}}Competitive\\ Equilbrium\end{tabular}} & Random-trigger\cite{DongBaroche-28,UllahPark-25}&\checkmark &$\times$ &\checkmark &\checkmark &$\times$\\
& Threshold-trigger\cite{LiuXu-32,LiLi-33} &\checkmark &$\times$ &\checkmark &$\times$ &\checkmark\\
& Extra Consensus\cite{Z.P.-29} &\checkmark& $\times$ &$\times$ &\checkmark &$\times$\\
& Node Coloring\cite{umer2021novel}     &\checkmark& $\times$ &\checkmark &\checkmark &\checkmark \\
& Our Approach     &\checkmark &\checkmark &\checkmark &\checkmark &\checkmark\\ \hline
\end{tabular}%
}\color{black}
\end{table}

In summary, this paper makes the following contributions:
\begin{itemize}[leftmargin=8pt]
\item \textit{Develop the connection-aware P2P trading algorithm}  to address the asynchrony challenge in P2P trading, taking into account the trading connection limits. Besides, this algorithm enables prosumers to autonomously select their trading partners during negotiations, fostering a more dynamic and user-centric trading environment. 
\item \textit{Propose smart selection strategies} within the context of connection limits in P2P trading. The smart selection strategies not only accelerate convergence in asynchronous settings, but also protect the autonomy of prosumers' selection choices to enhance individual welfare.
\item \textit{Provide theoretical guarantees} on the convergence of the connection-aware P2P trading algorithm. Additionally, it elucidates how smart selection strategies contribute to accelerating the convergence of the negotiation process.
\end{itemize}

The remainder of this paper is organized as follows: Section \ref{sec:standard} outlines the P2P trading framework, formulates the problem of computing the P2P trading equilibrium, and discusses the standard primal-dual solution algorithm. Section \ref{sec:partial} introduces two variants of the asynchronous P2P trading algorithm - the edge-based and node-based algorithms - and elaborates on the smart selection strategies. Section \ref{sec:convergence} presents a convergence proof for both algorithms and demonstrates the efficacy of the smart selection strategies in accelerating convergence. Section \ref{sec:case} conducts case studies validating the effectiveness of the proposed methods. Section \ref{sec:conclusions} offers conclusions and suggests avenues for future research.

\textit{Notations:} Bold italic $\boldsymbol{x}$ denote column vectors. $\mathrm{col}(\boldsymbol{x}_1,\boldsymbol{x}_2,..,)$ denotes the stacked column vector formed by $\boldsymbol{x}_1,\boldsymbol{x}_2,...$. $\left<\boldsymbol{x}_1,\boldsymbol{x}_2\right>$ denotes the inner product of $\boldsymbol{x}_1$ and $\boldsymbol{x}_2$. $\|\boldsymbol{x}\|$ denotes the L2-norm. $\nabla$ is the gradient operator. 

\section{Standard P2P Trading Model and Algorithm} \label{sec:standard}
\subsection{Trading Setting and Competitive Equilibrium Computation}
The system comprises prosumers, denoted as $i=1,...,I$. Each prosumer $i$ possesses DERs, including load equipment, rooftop photovoltaic panels, electric vehicles, and energy storage systems. These are integrated with an energy management system to facilitate monitoring, P2P trading, and billing. In the context of P2P trading, ponsumers $i$ are represented as nodes $\mathcal{N}$, and potential trading connections are symbolized by $\mathcal{E}$, following the conventions of graph theory. This study considers the situation where P2P trading is carried out in the day-ahead market, where $\tau=1,...,\bar{\tau}$ periods are considered. For prosumer $i$, its individual decision variables can include load usage power $\boldsymbol{p}_{i}^{\mathrm{L}}=\mathrm{col}(p_{i,\tau}^{\mathrm{L}})$, storage/vehicle charging power $\boldsymbol{p}_{i}^{\mathrm{CH}}=\mathrm{col}(p_{i,\tau}^{\mathrm{CH}})$, storage/vehicle discharging power $\boldsymbol{p}_{i}^{\mathrm{DIS}}=\mathrm{col}(p_{i,\tau}^{\mathrm{DIS}})$, storage/vehicle state of charge $\boldsymbol{s}_{i}=\mathrm{col}(s_{i,\tau})$, and exchanging power with the main grid $\boldsymbol{p}_{i}^{\mathrm{EX}}=\mathrm{col}(p_{i,\tau}^{\mathrm{EX}})$. For brevity, they are denoted in the concatenated form $\boldsymbol{x}_i$. $\boldsymbol{x}_i= \mathrm{col}(\boldsymbol{p}_{i}^{\mathrm{L}},\boldsymbol{p}_{i}^{\mathrm{CH}},\boldsymbol{p}_{i}^{\mathrm{DIS}},\boldsymbol{s}_{i},\boldsymbol{p}_{i}^{\mathrm{EX}})$. The information is private to prosumer $i$. P2P trading transactions between prosumer $i$ and $j$ are denoted as $\boldsymbol{t}_{i,j}=\mathrm{col}({t}_{i,j,\tau})$ for prosumer $i$ and $\boldsymbol{t}_{j,i}=\mathrm{col}({t}_{j,i,\tau})$ for prosumer $j$, which should only be known by $i$ and $j$. All possible P2P trading transactions of prosumer $i$ are denoted as $\boldsymbol{t}_{i}=\mathrm{col}(\boldsymbol{t}_{i,j})_{j\in\mathcal{N}_i}$. \textcolor{black}{A pair refers to two prosumers, $i$ and $j$, who have potential trading intentions. There may be non-zero electricity trading between them, along with the possibility of P2P communication links.}

\textcolor{black}{This paper primarily focuses on fully decentralized P2P trading negotiations among prosumers. In real-world P2P trading systems, the communication topology may be hybrid, where groups of prosumers communicate with a centralized coordinator, while individual members within each group communicate in a P2P manner. In this scenario, the proposed method can be employed to assist with the communication process within each prosumer group.}

\textcolor{black}{Prosumers typically have the ability to engage in strategic behavior and exercise bargaining power in markets. However, in P2P trading, where there are numerous buyers and sellers dealing in identical electricity products, no individual prosumer can influence the market price. As a result, the outcome converges to a \textit{competitive equilibrium}~\cite{LiLian-34}, which is the optimal solution to the social cost minimization problem, represented as:}
\begin{align}
	\min_{\boldsymbol{x}_i,\boldsymbol{t}_i} \,\,&\sum\nolimits_{i\in \mathcal{N} }^{}{J_i\left( \boldsymbol{x}_i,\boldsymbol{t}_i \right)} \label{eq:obj}
	\\
	(\bf{P1})~~~~\mathrm{s}.\mathrm{t}.~~&\left[ \boldsymbol{x}_i,\boldsymbol{t}_i \right] \in \Omega _i,~\forall i\in \mathcal{N} \label{eq:feasible}
	\\
	&\boldsymbol{t}_{i,j}+\boldsymbol{t}_{j,i}=\boldsymbol{0}: \boldsymbol{\lambda}_{i,j}, ~\forall i\in \mathcal{N},j\in \mathcal{N}_i \label{eq:balance}
\end{align}
where the objective function \eqref{eq:obj} is to minimize the sum of individual costs $J_i\left( \boldsymbol{x}_i,\boldsymbol{t}_i \right)$, which further depends on $i$'s decision variables $\boldsymbol{x}_i$ and P2P transactions $\boldsymbol{t}_i$. The cost typically includes negative load usage utility, storage age cost, network fee, tax and preferences for P2P trading, etc. Constraint \eqref{eq:feasible} is the operation constraint for prosumer $i$, denoted by $\Omega _i$. The constraints can include load shifting constraints, energy storage charging status constraints, and power output constraints, etc. All the related information belongs to prosumer $i$'s privacy. Constraint \eqref{eq:balance} is the P2P trading power balance constraint: the power exported from $i$ to $j$ ($\boldsymbol{t}_{i,j}$) should be balanced by power from $j$ from $i$ ($\boldsymbol{t}_{j,i}$). The corresponding dual variable is $\boldsymbol{\lambda}_{i,j}$, represented by $\boldsymbol{\lambda}_{i,j}=\mathrm{col}(\lambda_{i,j,\tau})$ in different trading periods. The algorithms and solution methods can be used for day-ahead, real-time and rolling-based P2P trading situations.

To enhance the generalizability of the proposed method, we do not specify the concrete expressions of the objective \eqref{eq:obj} and the individual constraint \eqref{eq:feasible}. Instead, certain critical characteristics of the problem are assumed (P1):
\begin{assumption}\label{convex}
	The individual constraint $\Omega _i$ is a convex set. The individual objective function $J_i\left( \boldsymbol{x}_i,\boldsymbol{t}_i \right)$ is convex with respect to $\boldsymbol{x}_i$ and $m_i$-strongly convex with respect to $\boldsymbol{t}_i$. The constraints  \eqref{eq:feasible}-\eqref{eq:balance} have at least one interior point.
\end{assumption}
This convexity assumption is common in modeling prosumers and is crucial for ensuring the optimality of P2P trading solutions. The strongly convex assumption reflects the marginally increasing tax or network fee relative to the P2P trading volume, as discussed in other literature~\cite{Nguyen-8,LiuXu-32,Z.P.-29,FengLiang-3}. Typically, when the unit tax or network fee is proportional to the P2P trading volume, approximating tiered tax rates, the total network tax or fee becomes a quadratic function of trading power, which exhibits strong convexity. Defining $\boldsymbol{\lambda}_{i}=\mathrm{col}(\boldsymbol{\lambda}_{i,j})_{j\in\mathcal{N}_i}$, we present the following theorem:
\begin{theorem}\label{th:equilbrium}
	Under assumption \ref{convex}, $(\boldsymbol{x}_{i}^{\star} ,\boldsymbol{t}_{i}^{\star}, \boldsymbol{\lambda }_{i}^{\star})$ is a competitive equilibrium if and only if $( \boldsymbol{x}_{i}^{\star},\boldsymbol{t}_{i}^{\star} ,\boldsymbol{\lambda }_{i}^{\star})$ is the optimal solution to the problem (P1)~\cite{LiLian-34}.
\end{theorem}
Theorem \ref{th:equilbrium} indicates that the P2P trading price can be viewed as the Lagrange multiplier of the trading balance constraint in problem (P1). Problem (P1) can also integrate network constraints, such as distribution power flow and voltage magnitude, into the optimization framework. For simplicity, this discussion focuses on the trading power balance constraint as a representative example.

\subsection{Primal-dual Algorithm for P2P trading}
\begin{algorithm}[t] 
	\caption{Standard P2P Trading}\label{al:standard}
	\SetAlgoLined
	\KwInput {Initial values $\boldsymbol{\lambda}_{i,j}^{\left( 0 \right)}  ,\boldsymbol{t}_{i}^{\left( 0 \right)}$; stopping criterion $\varepsilon_1,\varepsilon_2$; $k=1$}
	\While {$\|\Delta\boldsymbol{x}_i;\Delta\boldsymbol{t}_i\| \geqslant \varepsilon_1$ \textbf{or} $ \| \Delta \boldsymbol{\lambda }_{i,j}\| \geqslant \varepsilon_2$ }{
		\textbf{Step~1:} Each prosumer updates trade proposals  
		\begin{equation}
			\begin{aligned}
				&\boldsymbol{x}_{i}^{\left( k \right)},\boldsymbol{t}_{i}^{\left( k \right)}=
				\\
				&\mathrm{arg} \underset{\left[ \boldsymbol{x}_i,\boldsymbol{t}_i \right] \in \Omega _i}{\min}J_i\left( \boldsymbol{x}_i,\boldsymbol{t}_i \right) +\sum_{j\in \mathcal{N} _i}{\left< \boldsymbol{\lambda }_{i,j}^{\left( k-1 \right)},\boldsymbol{t}_{i,j} \right>}
			\end{aligned}
		\end{equation} 
		and transmit new proposals to \textit{all} trading peers.\\
		\textbf{Step~2:} Each prosumer receives new proposals and updates dual prices
		\begin{equation} \label{eq:standard_dual}
			\boldsymbol{\lambda }_{i,j}^{\left( k \right)}=\boldsymbol{\lambda }_{i,j}^{\left( k-1 \right)}+\rho \left( \boldsymbol{t}_{i,j}^{\left( k \right)}+\boldsymbol{t}_{j,i}^{\left( k \right)} \right) 
		\end{equation}		
		and sets $k\gets k+1$.   
	}
	\KwResult {$\boldsymbol{t}_i^{\left( k \right)},\boldsymbol{\lambda }_{i,j}^{\left( k \right)}$}
\end{algorithm}
Problem (P1) is typically resolved iteratively, often through primal-dual iterations, as detailed in Algorithm \ref{al:standard}. This iterative approach employs a subgradient algorithm on the dual function, converging to the optimal trading volume and pricing pair. The algorithm exhibits several key market properties:
\begin{itemize}[leftmargin=8pt]
	\item \textit{Welfare Maximization}: Social welfare is maximized at this equilibrium.
	\item \textit{Individual Rationality}: Every prosumer is better off by participating in P2P trading.
	\item \textit{Privacy Protection}: Each prosumer peer exchanges only information related to their mutual trade proposals.
\end{itemize}
\textcolor{black}{By privacy protection, we mean for prosumer pair $i$, he or she will share trading transactions $\boldsymbol{t}_{i,j}$ for potential trading partners $j$. It represents how much electricity will be traded between $i$ and $j$ according to the willingness of prosumer $i$.} Moreover, the primal-dual algorithm can be interpreted as prosumers engaging in a \textit{potential game}, where the potential function aligns with the objective function \eqref{eq:obj}. In certain scenarios, the individual cost function $J_i(\cdot)$ depends not only on $\boldsymbol{x}_i$ but also on the collective decisions of all prosumers $\boldsymbol{x}$. For example, the wholesale energy purchasing price may be influenced by the cumulative power import/ export of all prosumers. It makes the problem a \textit{Generalized Nash Equilibrium (GNE) seeking} problem~\cite{WangLiu-35}. Ensuring convergence in GNE seeking requires more stringent assumptions about problem (P1) and, in some cases, necessitates the pre-determination of P2P trading prices~\cite{AnandutaGrammatico-36}. 

\section{Connection-aware P2P Trading and Peer Selection} \label{sec:partial}
\begin{figure*}[t]
	\centering
	\includegraphics[width=0.8\textwidth]{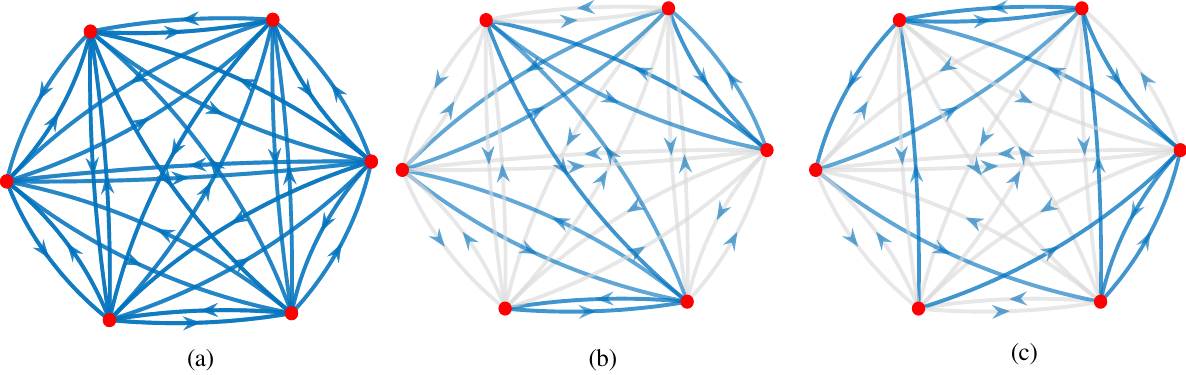}
	\caption{The compassion among communication topology for the three algorithms for 6 peers. The nodes represent prosumers. The colored edges represent activated trading connections, and the gray ones represent inactive edges. The arrow direction means message push direction. (a) Standard Algorithm requires strict push-receive synchronisation among all potential peers. (b) Edge-based Algorithm enables partial connections among all potential trading peers but still needs push-receive synchronisation (6 pair of edges are selected). (c) Node-based Algorithm enables both partial connections and de-synchronised push and response. Besides, prosumers can select prosumers to communication (2 connection limit for each prosumer).}
	\label{fig:intro}
\end{figure*}
Algorithm \ref{al:standard} has two primary shortcomings:
\begin{itemize}[leftmargin=8pt]
	\item \textit{Strict Synchronization}: Algorithm \ref{al:standard} requires stringent synchronicity among prosumers for updating trade proposals. Specifically, each prosumer $i$ must push updated trade proposals $\boldsymbol{t}_{i,j}^{(k)}$ to all potential trade peers $j\in \mathcal{N}i$ and concurrently await incoming proposals $\boldsymbol{t}_{j,i}^{(k)}$ from these peers. This is extremely difficult for realistic P2P trading. 
	\item \textit{Fixed and Compulsory Connection}: The algorithm pre-supposes a static trading network and every prosumer should always engage in mutual updates compulsorily. Contrary to this assumption, trading connections in real-world scenarios are often dynamically interconnected. Additionally, prosumers often have strong autonomy in selecting trading partners, necessitating a more flexible approach to modelling these connections. 
\end{itemize}
To address the aforementioned challenges, this section proposes two algorithms: edge-based connection-aware P2P trading and node-based connection-aware P2P trading. They aim to enhance the adaptability and efficiency of P2P trading by accommodating dynamic trading connections and reducing the dependency on strict synchronization among prosumers. Comparisons among the communication topology for the three algorithms are exhibited in Fig.\ref{fig:intro}.
\subsection{Edge-based connection-aware P2P Trading}
\begin{algorithm}[t] 
	\caption{Edge-based connection-aware P2P Trading}\label{al:partial_edge}
	\SetAlgoLined
	\KwInput {Initial values $\boldsymbol{\lambda}_{i,j}^{\left( 0 \right)}  ,\boldsymbol{t}_{i}^{\left( 0 \right)}$; stopping criterion $\varepsilon_1,\varepsilon_2$; $k=1$}
	\While {$\|\Delta\boldsymbol{x}_i;\Delta\boldsymbol{t}_i\| \geqslant \varepsilon_1$ \textbf{or} $ \| \Delta \boldsymbol{\lambda }_{i,j}\| \geqslant \varepsilon_2$ }{
		\textbf{Step~1:} Prosumers compute new trade proposals  
		\begin{align}
			&\boldsymbol{x}_{i}^{\left( k \right)},\boldsymbol{t}_{i}^{\left( k \right)}= 
			\\
			&\mathrm{arg} \underset{\left[ \boldsymbol{x}_i,\boldsymbol{t}_i \right] \in \Omega _i}{\min}J_i\left( \boldsymbol{x}_i,\boldsymbol{t}_i \right) +\sum_{j\in \mathcal{N} _i}{\left< \boldsymbol{\lambda }_{i,j}^{\left( k-1 \right)},\boldsymbol{t}_{i,j} \right>}\nonumber
		\end{align}\\
		\textbf{Step~2:} A sensory oracle determines trading peers $(i,j)\in\mathcal{E}^{(k)}$ to push proposal updates.\\
		\textbf{Step~3:} Each prosumer receives new proposals and updates dual prices if are selected:
		\begin{equation} \label{eq:edge_dual}
			\boldsymbol{\lambda }_{i,j}^{\left( k \right)}=\begin{cases}
				\boldsymbol{\lambda }_{i,j}^{\left( k-1 \right)}+\rho \left( \boldsymbol{t}_{i,j}^{\left( k \right)}+\boldsymbol{t}_{j,i}^{\left( k \right)} \right) , \left( i,j \right) \in \mathcal{E} ^{\left( k \right)}\\
				\boldsymbol{\lambda }_{i,j}^{\left( k-1 \right)}, \left( i,j \right) \notin \mathcal{E} ^{\left( k \right)}\\
			\end{cases}
		\end{equation}		
		Set $k\gets k+1$.   
	}
	\KwResult {$\boldsymbol{t}_i^{\left( k \right)},\boldsymbol{\lambda }_{i,j}^{\left( k \right)}$}
\end{algorithm}
The discussion begins with the edge-based connection-aware trading algorithm. Although it is more idealistic compared to the node-based connection-aware P2P trading algorithm, it remains crucial for understanding the primal-dual iteration process.

\begin{figure*}[t]
\centering
\includegraphics[width=0.7\textwidth]{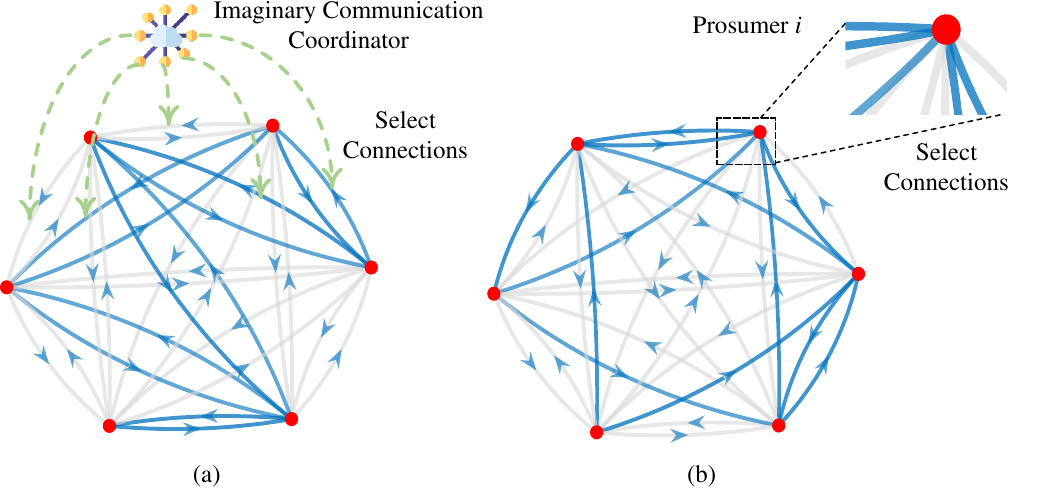}
\caption{(a) In Edge-based Algorithm, the sensory oracle selects which bilateral connections to be activated. (b) In Node-based Algorithm, each prosumer decides which potential trading partners to communicate with.}
\label{fig:oracle}
\end{figure*}

\color{black}
Algorithm \ref{al:partial_edge} represents a significant improvement over the standard P2P trading algorithm. A graphical illustration about the algorithm are illustrated in Fig.\ref{fig:oracle} (a). The key improvement of the algorithm lies in its relaxed requirement for updates from all trading peers. Specifically, in Step 2, a component referred to as the "imaginary communication coordinator" plays a crucial role. This coordinator is responsible for monitoring and controlling the P2P communication links, determining which trading connections $(i,j)\in\mathcal{E}^{(k)}$ should be activated during a given iteration $k$, thereby guiding the P2P proposal updates accordingly. For instance, the coordinator might use a random selection policy to decide which P2P communication link will be activated. This scenario reflects the real, time-varying communication network that dictates the feasible connections at any given moment. 

In Step 3, only the peers selected by the oracle are required to update their P2P trading prices. This approach introduces a form of semi-synchronization, contrasting with the Algorithm 1, which requires updates from all trading peers in every round. However, it maintains the principle that every message sent from the consumer $i$ to $j$ needs a corresponding response from $j$ to $i$, ensuring a closed communication loop. 
\color{black}

In this section, a simple illustration is provided to explain the effectiveness of Algorithm 2. The overall dual variables are defined as $\boldsymbol{\lambda}=\mathrm{col}({\boldsymbol{\lambda}_{i,j}})$. By the principle of strong duality for problem (P1), the following series of equations are derived:
\begin{align}
	&\min_{\boldsymbol{t}_{i,j}+\boldsymbol{t}_{j,i}=\boldsymbol{0},\left[ \boldsymbol{x}_i,\boldsymbol{t}_i \right] \in \Omega _i} \,\,\sum_i^{}{J_i\left( \boldsymbol{x}_i,\boldsymbol{t}_i \right)}\nonumber
	\\
	=&\underset{\boldsymbol{\lambda }}{\max}\min_{\left[ \boldsymbol{x}_i,\boldsymbol{t}_i \right] \in \Omega _i} \sum_{i}^{}{( J_i\left( \boldsymbol{x}_i,\boldsymbol{t}_i \right) +\sum_{j\in \mathcal{N} _i}{\left< \boldsymbol{\lambda }_{i,j},\boldsymbol{t}_{i,j} \right>} )}\label{eq:dual}
	\\
	=&-\min\nolimits_{\boldsymbol{\lambda }}\sum\nolimits_{i}^{}{-D_i\left( \boldsymbol{\lambda }_i \right)}=-\min\nolimits_{\boldsymbol{\lambda }}-D\left( \boldsymbol{\lambda } \right) \nonumber
\end{align}
where the individual dual function $D_i\left( \boldsymbol{\lambda }_i \right)$ and the summed dual function $D\left( \boldsymbol{\lambda } \right)$ is denoted as:
\begin{align}
	&D_i\left( \boldsymbol{\lambda }_i \right) =\min_{\left[ \boldsymbol{x}_i,\boldsymbol{t}_i \right] \in \Omega _i} J_i\left( \boldsymbol{x}_i,\boldsymbol{t}_i \right) +\sum\nolimits_{j\in \mathcal{N} _i}{\left< \boldsymbol{\lambda }_{i,j},\boldsymbol{t}_{i,j} \right>}
	\\
	&D\left( \boldsymbol{\lambda } \right) =\sum\nolimits_i{D_i\left( \boldsymbol{\lambda }_i \right)}
\end{align}
Meanwhile, the gradient of dual functions can be written as:
\begin{align} 
	&\nabla _{\boldsymbol{\lambda }_{i,j}}D_i(\boldsymbol{\lambda }_{i}^{\left( k-1 \right)})=\boldsymbol{t}_{i,j}^{\left( k \right)}
	\\
	&\nabla _{\boldsymbol{\lambda }_{i,j}}D(\boldsymbol{\lambda }^{\left( k-1 \right)})=\sum_i{\nabla _{\boldsymbol{\lambda }_{i,j}}D_i(\boldsymbol{\lambda }_{i}^{\left( k-1 \right)})}=\boldsymbol{t}_{i,j}^{\left( k \right)}+\boldsymbol{t}_{j,i}^{\left( k \right)}\label{eq:gradient}
\end{align}

These equations demonstrate that solving (P1) equates to finding the minimum of an unconstrained function $-D\left( \boldsymbol{\lambda } \right)$. The gradient descent method to achieve this is as follows:
\begin{align}
	\boldsymbol{\lambda }^{\left( k \right)}=\boldsymbol{\lambda }^{\left( k-1 \right)}+\rho \nabla D( \boldsymbol{\lambda }^{\left( k-1 \right)} ) 
\end{align}
This process corresponds to \eqref{eq:standard_dual} in Step 2 of Algorithm \ref{al:standard}. In contrast, the update rule \eqref{eq:edge_dual} in Step 2 of Algorithm \ref{al:partial_edge} implements a `coordinate descent' approach on the dual function of problem (P1)~\cite{HongWang-38}. This method involves selecting partial coordinates of $\boldsymbol{\lambda}$ for optimization, rather than following the full gradient direction. By ensuring that the negative dual function  $-D\left( \boldsymbol{\lambda } \right)$ consistently decreases with each iteration, convergence to the optimal value of (P1) is guaranteed. 

\subsection{Node-based Connection-Aware P2P Trading}
\begin{algorithm}[t] 
	\caption{Node-based Connection-Aware P2P Trading}\label{al:partial}
	\SetAlgoLined
	\KwInput {Initial values $\boldsymbol{\lambda}_{i,j}^{\left( 0 \right)}  ,\boldsymbol{t}_{i}^{\left( 0 \right)}$; stopping criterion $\varepsilon_1,\varepsilon_2$; $k=1$}
	\While {$\|\Delta\boldsymbol{x}_i;\Delta\boldsymbol{t}_i\| \geqslant \varepsilon_1$ \textbf{or} $ \| \Delta \boldsymbol{\lambda }_{i,j}\| \geqslant \varepsilon_2$ }{
		\textbf{Step~1:} Prosumers compute new trade proposals  
		\begin{align}
			&\boldsymbol{x}_{i}^{\left( k \right)},\boldsymbol{t}_{i}^{\left( k \right)}= 
			\\
			&\mathrm{arg} \underset{\left[ \boldsymbol{x}_i,\boldsymbol{t}_i \right] \in \Omega _i}{\min}J_i\left( \boldsymbol{x}_i,\boldsymbol{t}_i \right) +\sum_{j\in \mathcal{N} _i}{\left< \boldsymbol{\lambda }_{i,j}^{\left( k-1 \right)},\boldsymbol{t}_{i,j} \right>}\nonumber
		\end{align}\\
		\textbf{Step~2:} Prosumer $i$ selects the trading a subset of trading peers whom $i$ wants to communicate with, denoted by $j\in \mathcal{N} _i^{(k)}$, and pushes new trade proposals to $j\in\mathcal{N} _i^{(k)}$.\\
		\textbf{Step~3:} Each prosumer waits, receives new proposals, and records them as:
		\begin{equation}\label{eq:new_update}
			\tilde{\boldsymbol{t}}_{i,j}^{\left( k \right)}=\begin{dcases}
				\boldsymbol{t}_{i,j}^{\left( k \right)}, j\in \mathcal{N} _{i}^{(k)}\\
				\tilde{\boldsymbol{t}}_{i,j}^{\left( k-1 \right)}, j\notin \mathcal{N} _{i}^{(k)}\end{dcases},
			\tilde{\boldsymbol{t}}_{j,i}^{\left( k \right)}=\begin{cases}
				\boldsymbol{t}_{j,i}^{\left( k \right)}, i\in \mathcal{N} _{j}^{(k)}\\
				\tilde{\boldsymbol{t}}_{j,i}^{\left( k-1 \right)}, i\notin \mathcal{N} _{j}^{(k)}
			\end{cases}
		\end{equation}		 
		and update dual prices as:
		\begin{equation}\label{eq:new_dual}
			\boldsymbol{\lambda }_{i,j}^{\left( k \right)}=\begin{cases}
				\boldsymbol{\lambda }_{i,j}^{\left( k-1 \right)}+\rho \left( \tilde{\boldsymbol{t}}_{i,j}^{\left( k \right)}+\tilde{\boldsymbol{t}}_{j,i}^{\left( k \right)} \right) , i\,\mathrm{or}\, j\,\,\mathrm{updates}\\
				\boldsymbol{\lambda }_{i,j}^{\left( k-1 \right)}, \mathrm{otherwise}
			\end{cases}
		\end{equation}		
		and sets $k\gets k+1$.   
	}
	\KwResult {$\boldsymbol{t}_i^{\left( k \right)},\boldsymbol{\lambda }_{i,j}^{\left( k \right)}$}
\end{algorithm}
The node-based connection-aware P2P trading algorithm is shown in Algorithm \ref{al:partial}. It requires much less synchrony and grants prosumers more freedom, which is more realistic in real trading situations. A graphical illustration about the algorithm are illustrated in Fig.\ref{fig:oracle} (b). 

In Step 2, a prosumer $i$ selectively communicates with a subset of trading peers, denoted as $j\in \mathcal{N}_i^{(k)}$, to propose new trade agreements. This selective approach enables prosumer $i$ to actively manage the scope of their trading proposals. In practical implementations, this means prosumer $i$ can strategically choose the number of peers to communicate with, ensuring the timely transmission of information within the stipulated iteration deadline. This design allows for a broader network of potential trading partners, with negotiations limited to a manageable subset, chosen based on the prosumer’s preferences. Step 3 introduces a further level of flexibility by removing the obligation for prosumer $i$ to receive proposal updates from the selected peers $j\in \mathcal{N}_i^{(k)}$. The process of sending and receiving information is effectively de-synchronized. As outlined in Eq.\eqref{eq:new_update}, the trading proposal updates between prosumers $i$ and $j$ can occur in four scenarios: both parties propose new trades, only $i$ proposes, only $j$ proposes, or neither proposes. Importantly, as specified in Eq.\eqref{eq:new_dual}, the dual prices are updated whenever either prosumer, $i$ or $j$, initiates a new trade proposal, ensuring dynamic and responsive adjustments. 

\color{black}
Table \ref{tab:algorithm} compares the communication details of the original P2P trading algorithm with the proposed edge-based and node-based algorithms. Among these, the node-based algorithm is the most promising for P2P trading. It eliminates the need for bi-directional synchronization and allows prosumers to choose their communication links based on their preferences to trading partners.
\color{black}

\begin{table}[]
\color{black}
\centering
\caption{A comparison of communication details across different algorithms.}
\label{tab:algorithm}
\resizebox{0.8\textwidth}{!}{%
\begin{tabular}{ccccc}
\hline
\multicolumn{1}{l}{} & Shared Data                                                                       & \begin{tabular}[c]{@{}c@{}}Communication\\ Links\end{tabular} & \begin{tabular}[c]{@{}c@{}}Bi-directional\\ or Uni-directional\end{tabular} & \begin{tabular}[c]{@{}c@{}}Communication Link \\ Selection Method\end{tabular}        \\ \hline
Original Algorithm   & \multirow{3}{*}{\begin{tabular}[c]{@{}c@{}}P2P\\ transaction data\end{tabular}} & All links                                                     & Bi-directional                                                              & No selection                                                                          \\
Edge-based Algorithm &                                                                                   & Partial links                                                 & Bi-directional                                                              & \begin{tabular}[c]{@{}c@{}}By a imaginary \\ communication coordinator\end{tabular} \\
Node-based Algorithm &                                                                                   & Partial links                                                 & Uni-directional                                                             & By prosumers themselves                                                               \\ \hline
\end{tabular}%
}
\color{black}

\end{table}

\color{black}
In real-world scenarios, communication latency between prosumer pairs can fluctuate over time. Here, we analyze the impact of latency in practical implementations.

The original trading algorithm requires strict synchronization among prosumers. Each prosumer $i$ must first send messages to its neighbors (incurring sending latency) and then wait to receive information from every neighbor (incurring receiving latency). In large communities, where each prosumer may have numerous trading partners, the likelihood of encountering abnormal communication latency among neighbours increases. As a result, the total sending and receiving latency can become quite significant.

In the edge-based algorithm, the number of prosumers required to send messages is reduced, thereby decreasing the probability of encountering extreme communication latency. However, waiting deadlocks can still occur.

In the node-based algorithm, prosumers are no longer required to wait indefinitely for their neighbors' information. They can set a waiting deadline (e.g., 5 seconds), after which any incoming information is disregarded. This approach further reduces latency during the waiting phase.
\color{black}

\subsection{Smart Trading Peer Selection Strategy}
Selecting the appropriate trading peers for updates is a pivotal step in both Algorithm \ref{al:partial_edge} and Algorithm \ref{al:partial}. This selection process greatly influences the efficacy and speed of the iteration process. While conventional asynchronous algorithms typically resort to random peer selection, based on predefined probability distributions, a smarter approach is proposed that aligns with individual prosumer interests, and enhancing the convergence rate of the trading process.
\begin{itemize}[leftmargin=8pt]
	\item \textit{Smart Selection Strategy For Edge-based Algorithm}: For edge-based Algorithm \ref{al:partial_edge}, if the total number of activated trading connections is $\bar{e}$, the optimal strategy to accelerate convergence involves selecting the top-$e$ peers that yield the largest trading imbalances as: 
	\begin{equation}\label{eq:edge_selection}
		\underset{e_{i,j}=\left\{ 0,1 \right\} ,\sum\nolimits_{i,j}^{}{e_{i,j}=\bar{e}}}{\max}\,\,\sum_{i,j}{e_{i,j}\left\| \boldsymbol{t}_{i,j}^{\left( k \right)}+\boldsymbol{t}_{j,i}^{\left( k \right)} \right\|}
	\end{equation}
	\item \textit{Smart Selection Strategy For Node-based Algorithm}: For node-based Algorithm \ref{al:partial}, if the total number pieces of pushed information is $\bar{e}_i$ for $i$, the optimal strategy to accelerate convergence is to select the top-$\bar{e}_i$ peers that render the largest trading imbalance based on $i$'s information as:
	\begin{equation}\label{eq:node_selection}
		\underset{e_{i,j}=\left\{ 0,1 \right\} ,\sum\nolimits_i^{}{e_{i,j}=\bar{e}}}{\max}\,\,\sum_i{\sum_{j\in \mathcal{N} _i}{e_{i,j}\left\| \boldsymbol{t}_{i,j}^{\left( k \right)}+\tilde{\boldsymbol{t}}_{j,i}^{\left( k-1 \right)} \right\|}}
	\end{equation}
\end{itemize}

The economical intuitions behind this smart strategy lie in prioritizing peers with whom there are significant differences in trading proposals. These peers are more likely to influence the P2P trading price and, consequently, the cost for prosumer $i$. By focusing negotiation efforts on these potentially more impactful peers, prosumer $i$ can effectively accelerate the trading process. 

The two smart selection strategies presented above only requires a top-$e$ sorting process, whose complexity is far less than the solution process of prosumers' individual optimization problem. Thereby, it will not add extra complexity to the asychronous P2P trading process. 

In the following section, we aim to mathematically explain why these smart strategies can improve the convergence process. This proof will provide a solid theoretical foundation for the practical efficacy of the smart peer selection strategy. It is important to note, in the context of the real edge-based algorithm, that the ideal scenario of a sensory oracle selecting trading peers based on trading quantities does not exist in reality. Meanwhile, the node-based connection-aware algorithm is more realistic. 

%

\section{Convergence Analysis} \label{sec:convergence}
To begin, we examine certain properties of dual functions, which are fundamental to the convergence analysis. These properties, widely acknowledged in various studies (\cite{BianchiAnanduta-24,Y.Z.-37}), are as follows:
\begin{corollary}\label{th:Lipschitz}
	The individual dual function $D_i\left( \boldsymbol{\lambda }_i \right) $ has a $\frac{1}{m_i}$-Lipschitz gradient and the summed dual function $D\left( \boldsymbol{\lambda } \right) $ has $\sum_{i}{ \frac{1}{m_i} } $-Lipschitz gradient. Namely. for any two dual variables $\boldsymbol{\lambda}$ and $\boldsymbol{\mu}$, the following inequalities hold:
	\begin{align}
		&\frac{1}{m_i}\left\| \boldsymbol{\lambda }_i-\boldsymbol{\mu }_i \right\| \geqslant \left\| \nabla D_i\left( \boldsymbol{\lambda }_i \right) -\nabla D_i\left( \boldsymbol{\mu }_i \right) \right\| 
		\\
		&( \sum_{i}{ \frac{1}{m_i} } ) \left\| \boldsymbol{\lambda }-\boldsymbol{\mu } \right\| \geqslant \left\| \nabla D\left( \boldsymbol{\lambda } \right) -\nabla D\left( \boldsymbol{\mu } \right) \right\| \,\,             	
	\end{align}
\end{corollary}	
\begin{proof}
	See appendix.
\end{proof}
$L_i=\frac{1}{m_i}$ and $L=\sum_{i}{\frac{1}{m_i}}$ are denoted for brevity. To guarantee the convergence of Algorithms \ref{al:partial_edge} and \ref{al:partial}, it is necessary to ensure that all possible trading connections are activated periodically, preventing any trading connections from remaining hidden or inactive.
\begin{assumption}\label{finite_edge}
	All possible trading connections in Algorithm \ref{al:partial_edge} engage in mutual communications at least every $\bar{k}<\infty$ rounds. Similarly, in Algorithm \ref{al:partial}, each trading peer pushes updates to his neighbor $j$ within the $\bar{k}<\infty$ interval.
\end{assumption}
In the analysis below, it is also assumed that the dual function is well defined such that $\infty>-D( \boldsymbol{\lambda } )>-\infty$ for finite $\boldsymbol{\lambda }$.

\subsection{Edge-based Algorithm and its Smart Selection Strategy}
Based on Corollary \ref{th:Lipschitz}, the changes in the negative dual function $-D( \boldsymbol{\lambda } )$ can be bounded, as shown below:
\begin{prop}\label{prop:edge}
	In Algorithm \ref{al:partial_edge}, the negative dual function value at iteration $k$ satisfies:
	\begin{equation}\label{eq:edge}
		\begin{aligned}
			-D( \boldsymbol{\lambda }^{\left( k \right)} ) &\leqslant -D( \boldsymbol{\lambda }^{\left( k-1 \right)}) -
			\\
			&\sum\nolimits_{i,j\in \mathcal{E} ^{\left( k \right)}}^{}{\left( \frac{1}{\rho}-\frac{L}{2} \right) \left\| \boldsymbol{\lambda }_{i,j}^{\left( k \right)}-\boldsymbol{\lambda }_{i,j}^{\left( k-1 \right)} \right\| ^2}
		\end{aligned}	
	\end{equation}		
\end{prop}
\begin{proof}
	Because the dual function has $L$-Lipschitz gradients, the negative dual function $-D( \boldsymbol{\lambda })$ also has $L$-Lipschitz gradients. Thereby, the changes between two rounds can be bounded by:
	\begin{equation} \label{eq:lip}
		\begin{aligned}
			&-D( \boldsymbol{\lambda }^{\left( k \right)} ) \leqslant -D( \boldsymbol{\lambda }^{\left( k-1 \right)} )+ 
			\\
			& \langle -\nabla D( \boldsymbol{\lambda }^{\left( k-1 \right)} ) ,\boldsymbol{\lambda }^{\left( k \right)}-\boldsymbol{\lambda }^{\left( k-1 \right)}\rangle+\frac{L}{2}\| \boldsymbol{\lambda }^{\left( k \right)}-\boldsymbol{\lambda }^{\left( k-1 \right)}\| ^2
		\end{aligned}
	\end{equation}
	For peers $(i,j)\notin\mathcal{E}_k$ who do not push updates in round $k$, $\boldsymbol{\lambda }_{i,j}^{\left( k \right)}-\boldsymbol{\lambda }_{i,j}^{\left( k-1 \right)}=\boldsymbol{0}$. For peers $(i,j)\in\mathcal{E}_k$ who are selected to push updates, there is $\nabla_{\boldsymbol{\lambda}_{i,j}} D( \boldsymbol{\lambda }^{\left( k-1 \right)} )=\boldsymbol{t}_{i,j}^{\left( k \right)}+\boldsymbol{t}_{j,i}^{\left( k \right)}=\rho^{-1}(\boldsymbol{\lambda }_{i,j}^{\left( k \right)}-\boldsymbol{\lambda }_{i,j}^{\left( k-1 \right)})$.  Taking the expressions into the inequality \eqref{eq:lip} completes the proof.
\end{proof}	

The next theorem shows that if the learning rate $\rho$ is small, it is guaranteed that the gradient descent can always decrease the value of the negative dual function $-D\left( \boldsymbol{\lambda } \right)$ for Algorithm \ref{al:partial_edge}, which can finally lead to the convergence:
\begin{theorem} \label{th:edge}
	For Algorithm \ref{al:partial_edge}, under assumption \ref{al:partial_edge}, if the learning rate $\rho$ satisfies $0 < \rho < \frac{2}{L}$, then both algorithm converge to the optimal solution of problem (P1).
\end{theorem}
\begin{proof}
	Sum all the inequalities \ref{eq:edge} from $k'=0$ to $k'=k$:
	\begin{equation}
		\begin{aligned}
			&\sum_{k'=0}^{k'=k}{\sum_{i,j\in \mathcal{E} ^{( k' )}}^{}{\left( \frac{1}{\rho}-\frac{L}{2} \right) \left\| \boldsymbol{\lambda }_{i,j}^{( k' )}-\boldsymbol{\lambda }_{i,j}^{( k'-1 )} \right\| ^2}}\leqslant 
			\\
			&-D(\boldsymbol{\lambda }^{\left( 0 \right)})+D(\boldsymbol{\lambda }^{\left( k \right)})\leqslant -D(\boldsymbol{\lambda }^{\left( 0 \right)})+D(\boldsymbol{\lambda }^{\star})<\infty
		\end{aligned}
	\end{equation}
	where the second inequality is obtained from the definition of optimality of the dual function. By assumption 2, the term $\| \boldsymbol{\lambda }_{i,j}^{( k )}-\boldsymbol{\lambda }_{i,j}^{( k-1 )} \|$ appears infinitely many times for every $(i,j)\in\mathcal{E}$. Thereby, when $k\rightarrow \infty$, we must have $ \boldsymbol{\lambda }_{i,j}^{( k )}-\boldsymbol{\lambda }_{i,j}^{( k-1 )}\rightarrow \boldsymbol{0}$, thereby $\boldsymbol{t}_{i,j}^{\left( k \right)}+\boldsymbol{t}_{j,i}^{\left( k \right)}\rightarrow \boldsymbol{0}$ when $k\rightarrow \infty$. Combining the update rule of Step 1 in Algorithm \ref{al:partial_edge} and $\boldsymbol{t}_{i,j}^{\left( k \right)}+\boldsymbol{t}_{j,i}^{\left( k \right)}\rightarrow \boldsymbol{0}$, it is easy to see that the limit point is the optimal solution to (P1).
\end{proof}	
From inequality \eqref{eq:edge}, it is observed that by selecting trading peers who can induce larger gradient norms, the descending speed will be greedily accelerated. This concept forms the basis of the \textit{Gauss-Southwell update rule}: select the `best' coordinate rather than a random coordinate to perform gradient descents, which generally leads to a faster convergence rate~\cite{NutiniSchmidt-39}. Since the gradient is computed as Eq.\eqref{eq:gradient}, \textit{Gauss-Southwell} rule is translated into the smart strategy outlined in Eq.\eqref{eq:edge_selection}. 

\subsection{Node-based Algorithm and its Smart Selection Strategy}
In the context of Algorithm \ref{al:partial}, due to the asynchrony in Steps 2 and 3, Proposition \ref{prop:edge} cannot be applied directly. However, an upper bound can be established to account for the potential adverse effects of this asynchrony:
\begin{prop}\label{prop:node}
	In Algorithm \ref{al:partial}, the negative dual function value at iteration $k$ satisfies:
	\begin{equation}\label{eq:node}
		\begin{aligned}
			-D( \boldsymbol{\lambda }^{\left( k \right)}) &\leqslant -D( \boldsymbol{\lambda }^{\left( 0 \right)}) -
			\\
			\,\, &\left( \frac{1}{\rho}-\frac{L}{2}-2\bar{k}\bar{L}\right) \sum_k{\| \boldsymbol{\lambda }^{\left( k \right)}-\boldsymbol{\lambda }^{\left( k-1 \right)}\| ^2}
		\end{aligned}	
	\end{equation}	
	where $\bar{L} = \max_i L_i$ represents the maximum Lipschitz constant in individual dual functions.	
\end{prop}
\begin{proof}
	See appendix.
\end{proof}	
Based on Proposition \ref{prop:node}, the convergence guarantee for Algorithm \ref{th:node} can be obtained:
\begin{theorem} \label{th:node}
	For Algorithm \ref{al:partial}, under assumption \ref{al:partial_edge}, if the learning rate $\rho$ satisfies $ 0<\rho < \left( \frac{L}{2}+2\bar{k}\bar{L} \right) ^{-1} $, then the algorithm converges to the optimal solution of the problem (P1).
\end{theorem}
The proof follows a reasoning similar to that of Theorem \ref{th:edge}. This theorem implies that in scenarios involving asynchronous trading with numerous peers, a more conservative adjustment of prices (i.e., a smaller step size $\rho$) is prudent to ensure convergence.

Eq.\eqref{eq:node} suggests that selecting trading peers who can induce larger changes in dual variables can effectively accelerate the descent of $-D(\boldsymbol{\lambda})$. In the asynchronous setting of Algorithm \ref{al:partial}, each prosumer $i$ may not have prior knowledge of whether his peer $j\in\mathcal{N}i$ will push updates $\boldsymbol{t}_{j,i}^{(k)}$. Therefore, the optimal strategy for selecting trading peers is based on historical information $\tilde{\boldsymbol{t}}_{j,i}^{(k-1)}$, as outlined in the peer selection strategy Eq.\eqref{eq:node_selection}.

\section{Case Studies} \label{sec:case}
\subsection{Basic Settings and Optimal Trading Results}
\color{black} 
The proposed method can be applied to use cases across various time horizons, ranging from day-ahead (24 hours) to real-time (15 minutes). The case study will use a 24-hour profiled P2P trading example to demonstrate that the proposed method can handle more complex scenarios with longer optimization horizons. 
\color{black}

The P2P is carried out among a total of $I=10^3$ prosumers over the next $T=24$ hours. Each prosumer is randomly assigned potential trading peers, averaging $50$ peers per prosumer, resulting in a total of $40,927$ potential trading connections. The cost function $J_i\left( \boldsymbol{x}_i,\boldsymbol{t}_i \right)$ is set to be composed as follows:
$$
\begin{aligned}
	J_i\left( \boldsymbol{x}_i,\boldsymbol{t}_i \right) =&\sum\nolimits_{\tau}^{}{q _{\tau}^{\mathrm{b}}\left[ p_{i,\tau}^{\mathrm{EX}} \right] ^++q _{\tau}^{\mathrm{s}}\left[ p_{i,\tau}^{\mathrm{EX}} \right] ^-}-V_i\left( \boldsymbol{p}_{i}^{\mathrm{L}} \right) 
	\\
	&+C_{i}^{\mathrm{ES}}\left( \boldsymbol{p}_{i}^{\mathrm{CH}},\boldsymbol{p}_{i}^{\mathrm{DIS}} \right) +C_{i}^{\mathrm{P}2\mathrm{P}}\left( \boldsymbol{t}_{i} \right) 
\end{aligned}
$$ 
The expression is explained as follows:\\
1). $q _{\tau}^{\mathrm{b}}$ and $q _{\tau}^{\mathrm{s}}$ are the unit purchasing and selling prices in the wholesale market, respectively; $q_{\tau}^{\mathrm{b}}[p_{i,\tau}^{\mathrm{EX}}] ^+$ and $q _{\tau}^{\mathrm{s}}[p_{i,\tau}^{\mathrm{EX}}]^-$ denote the monetary amounts for purchasing/selling power in the wholesale market. Price data are derived from the daily average nodal price in the PJM market from July 2021 to July 2022. $\lambda _{t}^{\mathrm{b}}$ is set as twice the nodal price, and $\lambda _{t}^{\mathrm{s}}$ as $1.5$ times.\\
2). $V_i( \boldsymbol{p}_{i}^{\mathrm{L}} )$ is the prosumers' load usage utility function, which is defined as a quadratic function of load power $V_i( \boldsymbol{p}_{i}^{\mathrm{L}} ) =\sum_\tau\xi _{i,\tau} (p_{i,\tau}^{\mathrm{L}})^2+\varrho _{i,\tau} p_{i,\tau}^{\mathrm{L}}$. $I=10^3$ randomly selected load profiles from the Ireland CER project and the London LCL project are utilized to systematically generate prosumers' utility functions.\\ 
3). $C_{i}^{\mathrm{ES}}\left( \boldsymbol{p}_{i}^{\mathrm{CH}},\boldsymbol{p}_{i}^{\mathrm{DIS}} \right)$ is the storage aging cost. It is set as $C_{i}^{\mathrm{ES}}( \boldsymbol{P}_{i}^{\mathrm{CH}}+\boldsymbol{P}_{i}^{\mathrm{DIS}})=\sum_\tau c_i\times(P_{i,\tau}^{\mathrm{CH}}+P_{i,\tau}^{\mathrm{DIS}})$ where $c_i\in[2,4]$ ¢/kW is the unit aging cost.\\
4). P2P trading network usage fee:$C_{i}^{\mathrm{P}2\mathrm{P}}\left( \boldsymbol{t}_i \right) =\sum\nolimits_{j\in \mathcal{N} _i}^{}{\left( \alpha _{i,j}\left| \boldsymbol{t}_{i,j} \right| +\beta _{i,j} \right) \left| \boldsymbol{t}_{i,j} \right\|}$ is the P2P trading network usage fee. This fee captures the tierary tax scheme for P2P trading, where $\alpha _{i,j}=1$, $\beta _{i,j}=1$.

Individual constraints $\Omega_i$ include:\\
1. Power balance equation for the prosumer $i$ ($p_{i,\tau}^{\mathrm{PV}}$ is the generation power of PV panels):
$$
p_{i,\tau}^{\mathrm{EX}}=p_{i,\tau}^{\mathrm{L}}+p_{i,\tau}^{\mathrm{CH}}-p_{i,\tau}^{\mathrm{DIS}}+\sum\nolimits_{j\in\mathcal{N}_i}{t_{i,j,\tau}}-p_{i,\tau}^{\mathrm{PV}}
$$
2. Storage's SOC changes, continuity, and limits ($\eta _{i}^{\mathrm{CH}}$, $\eta _{i}^{\mathrm{DIS}}$ are the charging and discharging efficiency; $\underline{s}_i$,$\overline{s}_i$ are the SOC's lower and upper limits):
$$
\begin{aligned}
	&s_{i,\tau}=s_{i,\tau-1}+\eta _{i}^{\mathrm{CH}}p_{i,\tau}^{\mathrm{CH}}-\eta _{i}^{ \mathrm{DIS}}p_{i,\tau}^{\mathrm{DIS}}
	\\
	&s_{i,0}=s_{i,\bar{\tau}},  s_{i,\tau}\in \left[ \underline{s}_i,\overline{s}_i \right] 
\end{aligned}
$$
3. Load shift constraints ($ \underline{p_{i}^{\mathrm{L}}},\overline{p_{i}^{\mathrm{L}}}$ are the load power's lower and upper limits; $p_{i\Sigma}^{\mathrm{L}}$ is the minimum total daily loads):
$$
p_{i,\tau}^{\mathrm{L}}\in [ \underline{p_{i}^{\mathrm{L}}},\overline{p_{i}^{\mathrm{L}}} ] , \sum\nolimits_\tau{p_{i,\tau}^{\mathrm{L}}}\geqslant p_{i\Sigma}^{\mathrm{L}}
$$
4. Other limits ($\underline{p_{i,\tau}^{\mathrm{EX}}},\overline{p_{i,\tau}^{\mathrm{EX}}}$ are the lower and upper limits for exchanged power; $\overline{p_{i}^{\mathrm{CH}}},\overline{p_{i}^{\mathrm{DIS}}}$ are the upper limits for the storage's charging and discharging power):
$$
\begin{aligned}
	p_{i,\tau}^{\mathrm{EX}}\in[ \underline{p_{i,\tau}^{\mathrm{EX}}},\overline{p_{i,\tau}^{\mathrm{EX}}} ], p_{i,\tau}^{\mathrm{CH}}\in [ 0,\overline{p_{i}^{\mathrm{CH}}} ] , p_{i,\tau}^{\mathrm{DIS}}\in [ 0,\overline{p_{i}^{\mathrm{DIS}}} ] \\
\end{aligned}
$$
Some paramerters are set as follows: $\varrho_{i,t}$ is randomly selected between $[10,20]$ ¢/kW. $\xi _{i,t}$ is set to $-\varrho_{i,t}/2\overline{p_{i,t}^{\mathrm{L}}}$ ¢/(kW)$^2$. $\overline{p_{i,t}^{\mathrm{L}}}$ is set to be 3 times of the recorded load profile; $\underline{p_{i,t}^{\mathrm{L}}}$ is set to be half of the recorded load profile. The maximum capacity of energy storage is set to 4 times of the average daily load in recorded profiles. The SOC is limited between $[0.1,1]$ of its capacity. $s_{i,0}=s_{i,\bar{\tau}}$ is set to $0.55$ times of the storage capacity. The charging and discharging efficiencies are both set at $95\%$. For the primal-dual algorithm, the learning rate $\rho=0.5$. All initial values of variables are set to 0. All the experiments are implemented using Matlab software's parallel workers on the servers with an AMD EPYC 7H12 @ 2.60GHz CPU and 384.0 GB of RAM at the Beijing Super Cloud Computing Center. The optimization solver is the Gurobi software.

In analyzing the results obtained from solving problem (P1), as depicted in Fig.\ref{fig:result}. A key finding is the effective utilization of PV generation and energy storage facilitated by the P2P trading framework. Notably, it is observed that when P2P trading activity within the community is active, the corresponding trading prices typically fall between the wholesale market's purchasing and selling prices. This trend is a crucial reflection of the 'individual rationality' property inherent in the P2P trading scheme. This result underscores the potential of P2P trading to not only enhance energy efficiency at the community level but also to ensure that such enhancements are economically beneficial to each participant. 

\begin{figure}[t]
	\centering
	\includegraphics[width=0.6\textwidth]{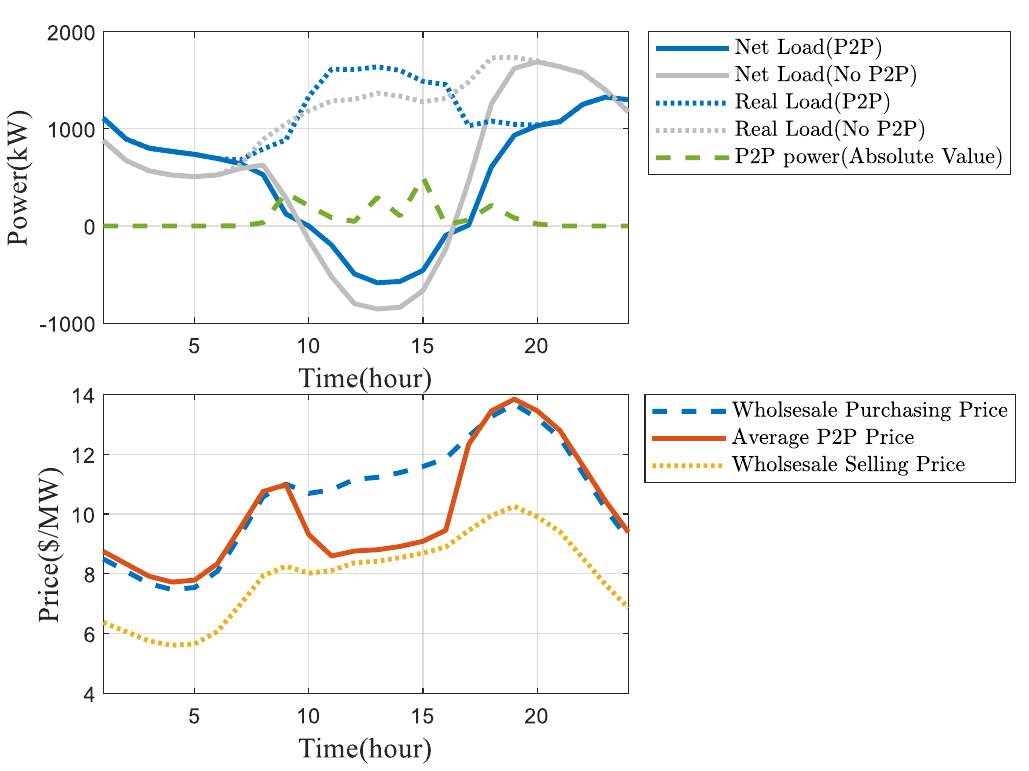}
	\caption{The optimal P2P trading solution and the P2P trading prices.}
	\label{fig:result}
\end{figure}

\subsection{Algorithms and Smart Selection Strategies}
\color{black}
The effectiveness of the smart selection strategies in the edge-based and node-based algorithms is evaluated, by employing the following benchmarks for comparison:
\begin{itemize}[leftmargin=8pt]
\item Round-Robin Strategy: In this strategy, the sensory oracle (for edge-based algorithms) or prosumer $i$ (for node-based algorithms) selects the trading peers in a sequential, cyclical manner. This ensures that each potential trading connection is selected in turn, before the cycle repeats. 
\item Random Strategy: In this strategy, the sensory oracle (in edge-based algorithms) or prosumer $i$ (in node-based algorithms) selects the trading peers in a random manner. This ensures each potential trading connection is selected in turn before the cycle repeats.
\item Smart Strategy: This strategy involves selection based on specific criteria outlined for edge-based algorithms and for node-based algorithms. Selections are made to optimize certain objectives, such as faster convergence or more efficient trading outcomes.
\end{itemize}
In the three communication link trading selection strategy, when the number of active communication links is fixed, the latency per iteration remains equal. Therefore, the total latency can be measured by the total number of iterations.
\color{black}
To compare these strategies, three indexes are used:
\begin{itemize}[leftmargin=8pt]
	\item Average Solution Optimality Gap: defined as $I^{-1}\sum\nolimits_i{\| \boldsymbol{x}_{i}^{\left( k \right)}-\boldsymbol{x}_{i}^{\star} \| +\| \tilde{\boldsymbol{t}}_{i}^{\left( k \right)}-\boldsymbol{t}_{i}^{\star} \|}$. This measures the average deviation of the solution in round $k$ from the optimal solution. 
	\item P2P Transaction Changes: defined as $I^{-1}\sum\nolimits_i{\| \tilde{\boldsymbol{t}}_{i}^{\left( k \right)}-\tilde{\boldsymbol{t}}_{i}^{\left( k-1 \right)} \|}$. Representing the primal residue, this index captures the changes in P2P transactions between consecutive rounds.
	\item P2P Price Changes: defined as $I^{-1}\sum\nolimits_i{\| \tilde{\boldsymbol{\lambda}}_{i}^{\left( k \right)}-\tilde{\boldsymbol{\lambda}}_{i}^{\left( k-1 \right)} \|}$. As the dual residue, this measures the fluctuation in P2P prices between consecutive rounds. 
\end{itemize}
It is important to note that primal and dual residues may not always be effective indicators of convergence, especially in asynchronous algorithms. In scenarios where only a few peers update in each round, the residue values can be misleadingly small, not necessarily reflecting true convergence.

\subsubsection{Edge-based Algorithm}
\begin{figure}[t]
	\centering
	\includegraphics[width=0.5\textwidth]{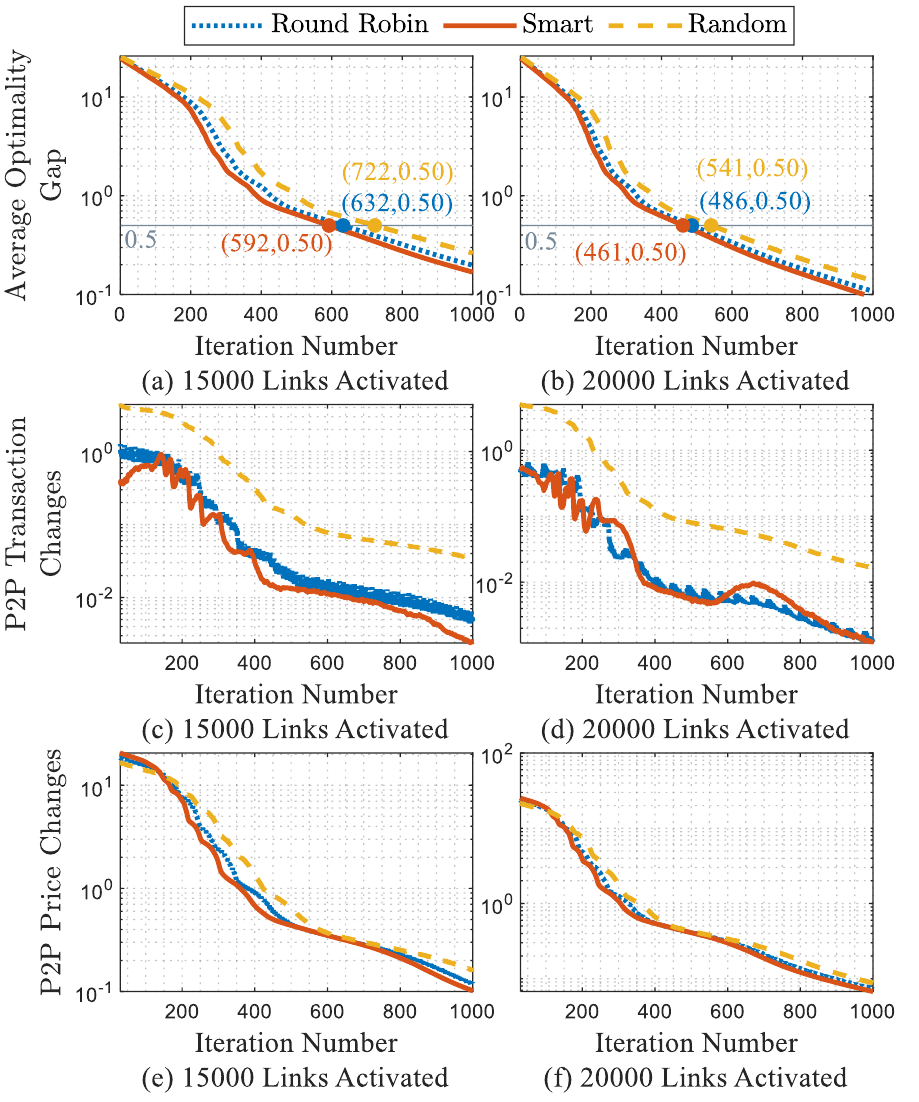}
	\caption{Convergence comparisons of (a-b) the average optimality gap, (c-d) the P2P transaction changes, and (e-f) the P2P Price changes for the edge-based connection-aware P2P trading algorithm.}
	\label{fig:edge}
\end{figure}
The effectiveness of the edge-based algorithm and its smart selection strategy is examined. It should be noted that the edge-based algorithm is a semisynchronized algorithm, since two prosumers of a trading pair should always perform mutual updates. Therefore, only two test cases are considered for the algorithm. Each case involves activating different proportions of peer connections per iteration round: approximately 37.5\% (15,000 edges) and 50\% (20,000 edges). 

The convergence results over 1,000 iteration rounds are illustrated in Fig.\ref{fig:edge}. Observations indicate that all three selection strategies—smart, random, and round-robin—progressively move towards convergence. The optimality gap is used as the principal criterion for assessing convergence. In particular, the smart strategy demonstrates a 6\%-7\% reduction in the number of iteration rounds required to achieve a gap level of 0.5, compared to random and round-robin strategies. The asynchrony in the algorithm manifests in the non-smooth decline of primal and dual residues. Despite this, it is evident that the smart selection strategy yields comparatively lower residue levels in the same iteration rounds. This suggests that while asynchrony introduces complexity into the convergence process, the smart strategy effectively mitigates this challenge, thus enhancing overall convergence efficiency.

\subsubsection{Node-based Algorithm}
\begin{figure*}[t]
	\centering
	\includegraphics[width=1\textwidth]{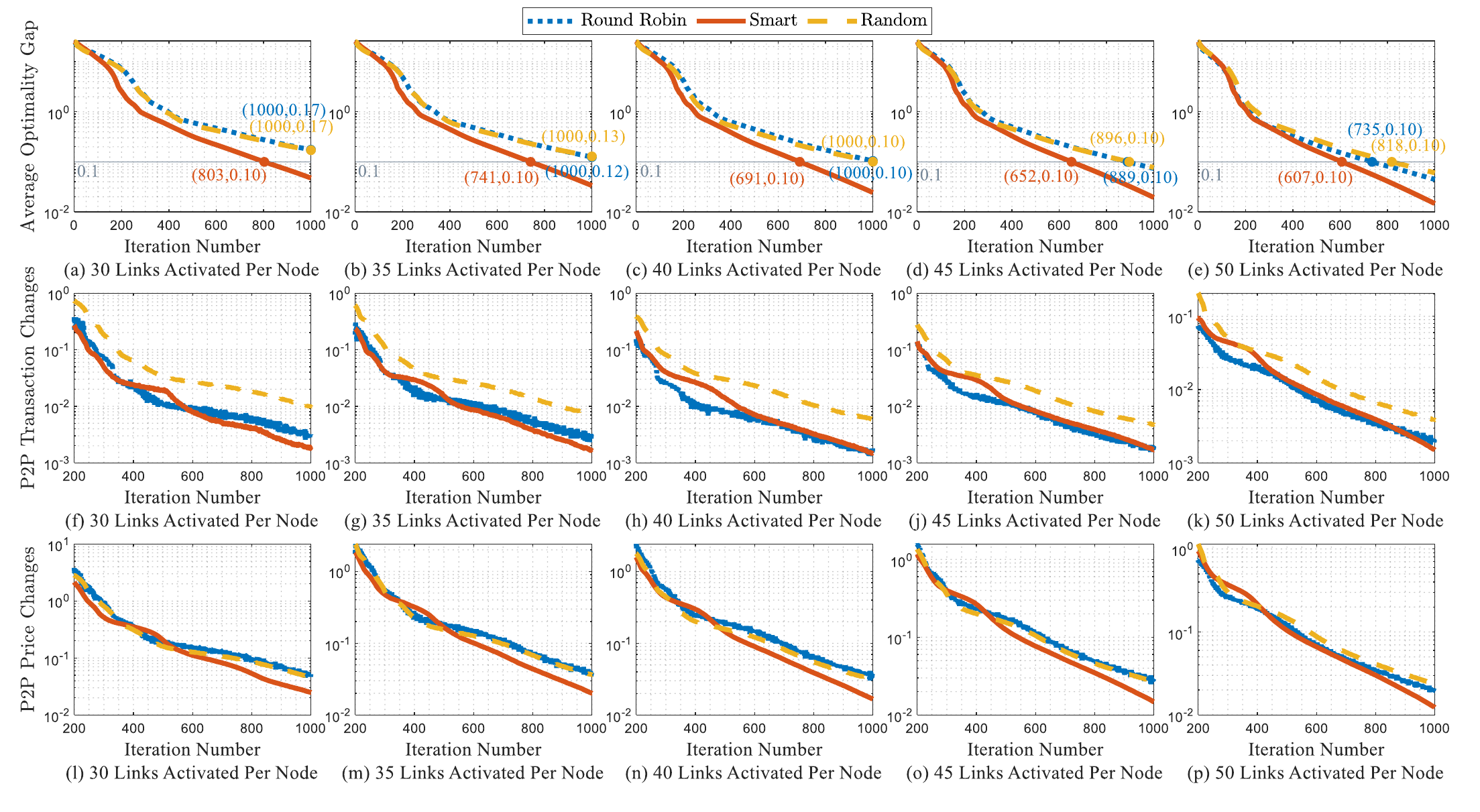}
	\caption{Convergence comparisons of (a-e) the average optimality gap, (f-k) the P2P transaction changes, and (l-p) the P2P Price changes for the node-based connection-aware P2P trading algorithm.}
	\label{fig:node}
\end{figure*}
Next, we examine the effectiveness of node-based algorithms and their smart selection strategy. Since the node-based algorithm is more realistic in real-world situations, a broader range of scenarios are examined. Specifically, five different cases are tested, varying the number of activated connection links per prosumer from 30 to 50 in increments of 5. The convergence results over 1,000 iteration rounds are depicted in Fig.\ref{fig:node}, which shows that all three selection strategies - smart, random and round-robin - are progressing towards convergence. Compared to edge-based algorithms, the smart peer selection strategy in node-based algorithms performs even better. It can reduce the number of iteration rounds needed to reach the level of the 0.1 gap by $15\%$ to $35\%$ compared to random and round-robin strategies. This analysis highlights that, through well-devised peer selection, node-based algorithms can achieve rapid convergence, making them highly suitable for real-world P2P energy trading systems.

\section{Conclusion} \label{sec:conclusions}
In this paper, a connection-aware trading algorithms and smart selection strategies designed for large-scale prosumer P2P trading networks is proposed. Our approach effectively addresses the challenges of asynchrony in P2P trading, while fully honoring the autonomy of prosumers in their self-selection strategies.  For the realistic node-based connection-aware P2P trading algorithm, the proposed method can not only lead to convergence but also to a reduction in trading agreement convergence time by $15\%$ to $35\%$. This represents a substantial improvement in efficiency, paving the way for more dynamic and responsive P2P trading environments. Future work can include a more detailed model for network constraints and a broader range of prosumer details. Additionally, the framework may be adaptable to other types of information structures of P2P trading and could include more diverse types of algorithms.

\appendix

\section*{Appendix}
\subsection{Proof of the Corollary \ref{th:Lipschitz}}
When the dual variables are $\boldsymbol{\lambda}$ and  $\boldsymbol{\mu}$, the optimal solution to prosumer $i$'s optimization problems is denoted by:
\begin{align}
	&\boldsymbol{x}_{i}',\boldsymbol{t}_{i}'=\mathrm{arg}\min\nolimits_{\left[ \boldsymbol{x}_i,\boldsymbol{t} \right] \in \Omega _i} J_i\left( \boldsymbol{x}_i,\boldsymbol{t}_i \right) +\left< \boldsymbol{\lambda }_i,\boldsymbol{t}_i \right> 
	\\
	&\boldsymbol{x}_{i}'',\boldsymbol{t}_{i}''=\mathrm{arg}\min\nolimits_{\left[ \boldsymbol{x}_i,\boldsymbol{t} \right] \in \Omega _i} J_i\left( \boldsymbol{x}_i,\boldsymbol{t}_i \right) +\left< \boldsymbol{\mu }_i,\boldsymbol{t}_i \right> 
\end{align}
According to the optimality condition of the constrained optimization problem, we have:
\begin{align}
	&\left< \nabla _{\boldsymbol{x}_i}J_i\left( \boldsymbol{x}_{i}',\boldsymbol{t}_{i}' \right) ,\boldsymbol{x}_{i}''-\boldsymbol{x}_{i}' \right> \nonumber
	\\
	&~~~~~~~~+\left< \nabla _{\boldsymbol{t}_{\boldsymbol{i}}}J_i\left( \boldsymbol{x}_{i}',\boldsymbol{t}_{i}' \right) +\boldsymbol{\lambda }_i,\boldsymbol{t}_{i}''-\boldsymbol{t}_{i}' \right> \geqslant 0
	\\
	&\left< \nabla _{\boldsymbol{x}_i}J_i\left( \boldsymbol{x}_{i}'',\boldsymbol{t}_{i}'' \right) ,\boldsymbol{x}_{i}'-\boldsymbol{x}_{i}'' \right> \nonumber
	\\
	&~~~~~~~~+\left< \nabla _{\boldsymbol{t}_{\boldsymbol{i}}}J_i\left( \boldsymbol{x}_{i}'',\boldsymbol{t}_{i}'' \right) +\boldsymbol{\mu }_i,\boldsymbol{t}_{i}'-\boldsymbol{t}_{i}'' \right> \geqslant 0
\end{align}
Summing the two inequalities, we have:
\begin{align}
	&\left< \boldsymbol{\lambda }_i-\boldsymbol{\mu }_i,\boldsymbol{t}_{i}'-\boldsymbol{t}_{i}'' \right> \geqslant \nonumber
	\\
	&~~~\left< \nabla _{\boldsymbol{x}_i}J_i\left( \boldsymbol{x}_{i}',\boldsymbol{t}_{i}' \right) -\nabla _{\boldsymbol{x}_i}J_i\left( \boldsymbol{x}_{i}'',\boldsymbol{t}_{i}'' \right) ,\boldsymbol{x}_{i}'-\boldsymbol{x}_{i}'' \right> + \label{eq:sum}
	\\
	&~~~\left< \nabla _{\boldsymbol{t}_{\boldsymbol{i}}}J_i\left( \boldsymbol{x}_{i}',\boldsymbol{t}_{i}' \right) -\nabla _{\boldsymbol{t}_{\boldsymbol{i}}}J_i\left( \boldsymbol{x}_{i}'',\boldsymbol{t}_{i}'' \right) ,\boldsymbol{t}_{i}'-\boldsymbol{t}_{i}'' \right> \nonumber
\end{align}
According to assumption \ref{convex}, we further have:
\begin{align}
	&\left< \nabla _{\boldsymbol{x}_i}J_i\left( \boldsymbol{x}_{i}',\boldsymbol{t}_{i}' \right) -\nabla _{\boldsymbol{x}_i}J_i\left( \boldsymbol{x}_{i}'',\boldsymbol{t}_{i}'' \right) ,\boldsymbol{x}_{i}'-\boldsymbol{x}_{i}'' \right>  + 
	\\
	&\left< \nabla _{\boldsymbol{t}_{\boldsymbol{i}}}J_i\left( \boldsymbol{x}_{i}',\boldsymbol{t}_{i}' \right) -\nabla _{\boldsymbol{t}_{\boldsymbol{i}}}J_i\left( \boldsymbol{x}_{i}'',\boldsymbol{t}_{i}'' \right) ,\boldsymbol{t}_{i}'-\boldsymbol{t}_{i}'' \right> \geqslant m_i\left\| \boldsymbol{t}_{i}'-\boldsymbol{t}_{i}'' \right\| ^2 \nonumber
\end{align}
Combining \eqref{eq:sum}, we have:
\begin{equation} \label{eq0}
	\begin{aligned}
		&\left\| \boldsymbol{\lambda }_i-\boldsymbol{\mu }_i \right\| \left\| \boldsymbol{t}_{i}'-\boldsymbol{t}_{i}'' \right\| \geqslant \left< \boldsymbol{\lambda }_i-\boldsymbol{\mu }_i,\boldsymbol{t}_{i}'-\boldsymbol{t}_{i}'' \right>  \geqslant m_i\left\| \boldsymbol{t}_{i}'-\boldsymbol{t}_{i}'' \right\| ^2
	\end{aligned}
\end{equation}
Using $\nabla D_i\left( \boldsymbol{\lambda }_i \right) =\boldsymbol{t}_{i}',\nabla D_i\left( \boldsymbol{\mu }_i \right) =\boldsymbol{t}_{i}''$, we get:
\begin{equation} \label{eq:lip1}
	m_i^{-1}\left\| \boldsymbol{\lambda }_i-\boldsymbol{\mu }_i \right\| \geqslant  \left\| \nabla D_i\left( \boldsymbol{\lambda }_i \right) -\nabla D_i\left( \boldsymbol{\mu }_i \right) \right\| 
\end{equation}
Thereby, $D_i\left( \boldsymbol{\lambda }_i\right)$ has a $\frac{1}{m_i}$-Lipschitz gradient. Moreover, we can write the inequalitie \eqref{eq:lip1} as:
\begin{equation} \label{eq:lip2}
	m_i^{-1}\left\| \boldsymbol{\lambda }-\boldsymbol{\mu } \right\| \geqslant  \left\| \nabla D_i\left( \boldsymbol{\lambda } \right) -\nabla D_i\left( \boldsymbol{\mu } \right) \right\| 
\end{equation}
because $\nabla_{\boldsymbol{\lambda}_{i',j'}} D_i\left( \boldsymbol{\lambda } \right)= \boldsymbol{0}$ if $(i',j')$ does not do not appear in $i$'s optimization problem. According to the definition of  $D\left( \boldsymbol{\lambda }\right)$, we have:
\begin{equation}
	\begin{aligned}
		&\left\| \nabla D\left( \boldsymbol{\lambda } \right) -\nabla D\left( \boldsymbol{\mu } \right) \right\| =\left\| \sum\nolimits_i^{}{\nabla D_i\left( \boldsymbol{\lambda } \right) -\nabla D_i\left( \boldsymbol{\mu } \right)} \right\| 
		\\
		&\overset{\left( a \right)}{\leqslant}\sum\nolimits_i^{}{\left\| \nabla D_i\left( \boldsymbol{\lambda } \right) -\nabla D_i\left( \boldsymbol{\mu } \right) \right\|}\overset{\left( b \right)}{\leqslant} ( \sum\nolimits_{i}{ m_i^{-1} } ) \left\| \boldsymbol{\lambda }-\boldsymbol{\mu } \right\|                
	\end{aligned}
\end{equation}
where (a) we use the triangle inequality; (b) we use the inequality \eqref{eq:lip2}.
\subsection{Proof of the Proposition \ref{th:node}}
The following corollary is provided for proof assistance:
\begin{corollary}\label{th:difference}
	The difference between newest trading proposal $\boldsymbol{t}_{i}^{(k)}$ and $\tilde{\boldsymbol{t}}_{i}^{(k)}$ can be bounded as follows:
	\begin{equation}\label{eq1}
		\| \boldsymbol{t}_{i}^{\left( k \right)}-\tilde{\boldsymbol{t}}_{i}^{(k)} \| \leqslant \sum_{k'=k-\bar{k}}^{k'=k}{\frac{1}{L_i}\| \boldsymbol{\lambda }_{i}^{( k'+1)}-\boldsymbol{\lambda }_{i}^{( k' )} \|}
	\end{equation}
\end{corollary}
\begin{proof}
	Denote the latest round that prosumer $i$ pushes updates to peer $j$ as $k_{i,j}^{i}$. Because each $\boldsymbol{t}_{i,j}$ is at least updated in the last $\bar{k}$ rounds, define the following anxillary variable:
	\begin{equation}
		\hat{\boldsymbol{t}}_{i,j}^{(k')}=\begin{cases}
			\boldsymbol{t}_{i,j}^{(k')}, k\geqslant k'\geqslant k_{i,j}^{i}\\
			\boldsymbol{t}_{i,j}^{(k_{i,j}^{i})}, k_{i,j}^{i}>k'\geqslant k-\bar{k}\\
		\end{cases}
	\end{equation}
	$\hat{\boldsymbol{t}}_{i,j}^{(k')}$ equals to $\boldsymbol{t}_{i,j}^{(k')}$ after the most updates, and remains the same value $\boldsymbol{t}_{i,j}^{(k_{i,j}^{i})}$ before updates. From the definition, it is easy to obtain:
	\begin{equation}
		\left\| \hat{\boldsymbol{t}}_{i,j}^{(k'+1)}-\hat{\boldsymbol{t}}_{i,j}^{(k')} \right\| \leqslant \left\| \boldsymbol{t}_{i,j}^{(k'+1)}-\boldsymbol{t}_{i,j}^{(k')} \right\|, \forall k \geqslant k'\geqslant k-\bar{k} 
	\end{equation}
	Thereby, for $\hat{\boldsymbol{t}}_{i}=\mathrm{col}(\hat{\boldsymbol{t}}_{i,j})$ and $\boldsymbol{t}_{i}=\mathrm{col}(\boldsymbol{t}_{i,j})$, we have
	\begin{equation} \label{eq3}
		\left\| \hat{\boldsymbol{t}}_{i}^{(k'+1)}-\hat{\boldsymbol{t}}_{i}^{(k')} \right\| \leqslant \left\| \boldsymbol{t}_{i}^{(k'+1)}-\boldsymbol{t}_{i}^{(k')} \right\|, \forall k \geqslant k'\geqslant k-\bar{k} 
	\end{equation}
	Then inequality \eqref{eq1} can be obtained as follows:
	\begin{align}
		&\| \boldsymbol{t}_{i}^{\left( k \right)}-\tilde{\boldsymbol{t}}_{i}^{(k)} \| 
		\\
		=&\left\| \mathrm{col}\left(\boldsymbol{t}_{i,j}^{\left( k \right)}-\boldsymbol{t}_{i,j}^{(k_{i,j}^{i})}\right) \right\|=\left\|\mathrm{col}\left(\sum_{k'=k_{i,j}^{i}}^{k'=k-1}{\boldsymbol{t}_{i,j}^{( k'+1 )}-\boldsymbol{t}_{i,j}^{( k' )}}\right) \right\| \nonumber
		\\
		\overset{\left( a \right)}{=}&\left\| \mathrm{col}\left( \sum_{k'=\bar{k}}^{k'=k-1}{\hat{\boldsymbol{t}}_{i,j}^{(k'+1)}-\hat{\boldsymbol{t}}_{i,j}^{(k')}} \right) \right\| =\left\| \sum_{k'=\bar{k}}^{k'=k-1}{\hat{\boldsymbol{t}}_{i}^{(k'+1)}-\hat{\boldsymbol{t}}_{i}^{(k')}} \right\| \nonumber
		\\
		\overset{\left( b \right)}{\leqslant}& \sum_{k'=k-\bar{k}}^{k'=k-1}{\left\| \hat{\boldsymbol{t}}_{i}^{( k'+1)}-\hat{\boldsymbol{t}}_{i}^{( k' )} \right\|} \overset{\left( c \right)}{\leqslant} \sum_{k'=k-\bar{k}}^{k'=k-1}{\left\| \boldsymbol{t}_{i}^{( k'+1)}-\boldsymbol{t}_{i}^{( k' )} \right\|} \nonumber
	\end{align}
	where (a) we use the fact that $\hat{\boldsymbol{t}}_{i,j}^{(k'+1)}-\hat{\boldsymbol{t}}_{i,j}^{(k')}=\boldsymbol{0}$ for $k_{i,j}^{i}>k'\geqslant k-\bar{k}$. In (b), we use the triangle inequality. In (c), we use \eqref{eq3}. Finally, combining \eqref{eq0}, the proof is completed.
\end{proof}
Now we are ready to prove Propostion \ref{th:node}.
\begin{proof}
	The bound \eqref{eq:lip} still holds, rewritten as:
	\begin{align} \label{eq4}
		-D&( \boldsymbol{\lambda }^{\left( k \right)} ) \leqslant -D( \boldsymbol{\lambda }^{\left( k-1 \right)} ) +\frac{L}{2}\sum\nolimits_{\left( i,j \right) \in \mathcal{E}}{\| \boldsymbol{\lambda }_{i,j}^{\left( k \right)}-\boldsymbol{\lambda }_{i,j}^{\left( k-1 \right)}\|} ^2+ \nonumber
		\\
		&\sum\nolimits_{\left( i,j \right) \in \mathcal{E}}{\left< -\nabla _{\boldsymbol{\lambda }_{i,j}^{}}D( \boldsymbol{\lambda }^{\left( k-1 \right)} ) ,\boldsymbol{\lambda }_{i,j}^{\left( k \right)}-\boldsymbol{\lambda }_{i,j}^{\left( k-1 \right)}\right>} 
	\end{align}
	The complicated situation is when only a single prosumer in one trading peer push updates. The last term on the right-hand side of \eqref{eq4} is re-written as:
	\begin{align}
		\,\,  &-\sum\nolimits_{\left( i,j \right) \in \mathcal{E}}\left< \nabla _{\boldsymbol{\lambda }_{i,j}^{}}D( \boldsymbol{\lambda }^{\left( k-1 \right)} ) ,\boldsymbol{\lambda }_{i,j}^{\left( k \right)}-\boldsymbol{\lambda }_{i,j}^{\left( k-1 \right)}\right> \nonumber
		\\
		&=-\sum\nolimits_{( i,j ) \in \mathcal{E}}\rho^{-1}\|  \boldsymbol{\lambda }_{i,j}^{\left( k \right)}-\boldsymbol{\lambda }_{i,j}^{\left( k-1 \right)}\|^2 + \sum\nolimits_{\left( i,j \right) \in \mathcal{E}} \label{eq:eq2}
		\\
		&\left< -\nabla _{\boldsymbol{\lambda }_{i,j}^{}}D( \boldsymbol{\lambda }^{\left( k-1 \right)}) +\rho^{-1}(\boldsymbol{\lambda }_{i,j}^{\left( k \right)}-\boldsymbol{\lambda }_{i,j}^{\left( k-1 \right)}),\boldsymbol{\lambda }_{i,j}^{\left( k \right)}-\boldsymbol{\lambda }_{i,j}^{\left( k-1 \right)}\right> \nonumber 
	\end{align}
	We next provide a worst case bound of the last asychoronous term on the right-hand side of the inequality \eqref{eq:eq2}:
	\begin{align}
		&\sum\nolimits_{( i,j ) \in \mathcal{E}} \left< -\left( \boldsymbol{t}_{i,j}^{\left( k \right)}+\boldsymbol{t}_{j,i}^{\left( k \right)} \right) +\left( \tilde{\boldsymbol{t}}_{i,j}^{\left( k \right)}+\tilde{\boldsymbol{t}}_{j,i}^{\left( k \right)} \right) ,\boldsymbol{\lambda }_{i,j}^{\left( k \right)}-\boldsymbol{\lambda }_{i,j}^{\left( k-1 \right)}\right> \nonumber
		\\
		=&-\sum\nolimits_{i \in \mathcal{N}}\left< \boldsymbol{t}_{i}^{\left( k \right)}-\tilde{\boldsymbol{t}}_{i}^{\left( k \right)},\boldsymbol{\lambda }_{i}^{\left( k \right)}-\boldsymbol{\lambda }_{i}^{\left( k-1 \right)}\right>  
		\\
		\leqslant& \sum\nolimits_{i \in \mathcal{N}}\left\| \boldsymbol{t}_{i}^{\left( k \right)}-\tilde{\boldsymbol{t}}_{i}^{\left( k\right)} \right\| \left\| \boldsymbol{\lambda }_{i}^{\left( k \right)}-\boldsymbol{\lambda }_{i}^{\left( k-1 \right)} \right\| \nonumber
	\end{align}
	It can be further bounded as follows:
	\begin{align}
		&\,\,  \left\| \boldsymbol{t}_{i}^{\left( k \right)}-\tilde{\boldsymbol{t}}_{i}^{\left( k \right)} \right\| \left\| \boldsymbol{\lambda }_{i}^{\left( k \right)}-\boldsymbol{\lambda }_{i}^{\left( k-1 \right)} \right\|  \nonumber
		\\
		\overset{\left( a \right)}{\leqslant} &\sum_{k'=k-\bar{k}}^{k'=k-1}{L_i\left\| \boldsymbol{\lambda }_{i}^{(k'+1)}-\boldsymbol{\lambda }_{i}^{(k')} \right\|}\left\| \boldsymbol{\lambda }_{i}^{\left( k \right)}-\boldsymbol{\lambda }_{i}^{\left( k-1 \right)} \right\|  \label{bound1}
		\\
		\overset{\left( b \right)}{\leqslant} &\frac{L_i}{2}\sum_{k'=k-\bar{k}}^{k'=k-1}{\left\| \boldsymbol{\lambda }_{i}^{\left( k \right)}-\boldsymbol{\lambda }_{i}^{\left( k-1 \right)} \right\| ^2+\left\| \boldsymbol{\lambda }_{i}^{(k'+1)}-\boldsymbol{\lambda }_{i}^{(k')} \right\| ^2}  \nonumber
		\\
		\leqslant &\frac{L_i\bar{k}}{2}\left\| \boldsymbol{\lambda }_{i}^{\left( k \right)}-\boldsymbol{\lambda }_{i}^{\left( k-1 \right)} \right\| ^2+\frac{L_i}{2}\sum_{k'=k-\bar{k}}^{k'=k-1}{\left\| \boldsymbol{\lambda }_{i}^{(k'+1)}-\boldsymbol{\lambda }_{i}^{(k')} \right\| ^2}  \nonumber
	\end{align}
	where (a) we use Corollory \ref{th:difference}. In (b), we use Young's inequality. Summing all the inequalities \eqref{bound1} for all $i$:
	\begin{align}
		&\,\, \sum\nolimits_{i\in \mathcal{N}}^{}{\left\| \boldsymbol{t}_{i}^{\left( k \right)}-\tilde{\boldsymbol{t}}_{i}^{\left( k \right)} \right\| \left\| \boldsymbol{\lambda }_{i}^{\left( k \right)}-\boldsymbol{\lambda }_{i}^{\left( k-1 \right)} \right\|}\nonumber
		\\
		&\leqslant \sum_{i\in \mathcal{N}}^{}{\frac{L_i\bar{k}}{2}\left\| \boldsymbol{\lambda }_{i}^{\left( k \right)}-\boldsymbol{\lambda }_{i}^{\left( k-1 \right)} \right\| ^2+\frac{L_i}{2}\sum_{k'=k-\bar{k}}^{k'=k-1}{\left\| \boldsymbol{\lambda }_{i}^{(k'+1)}-\boldsymbol{\lambda }_{i}^{(k')} \right\| ^2}} \nonumber
		\\
		&\overset{\left( a \right)}{=}\sum\nolimits_{\left( i,j \right) \in \mathcal{E}}^{}{\frac{\left( L_i+L_j \right) \bar{k}}{2}\left\| \boldsymbol{\lambda }_{i,j}^{\left( k \right)}-\boldsymbol{\lambda }_{i,j}^{\left( k-1 \right)} \right\| ^2} 
		\\
		&~~~+\sum\nolimits_{\left( i,j \right) \in \mathcal{E}}^{}{\frac{L_i+L_j}{2}\sum_{k'=k-\bar{k}}^{k'=k-1}{\left\| \boldsymbol{\lambda }_{i,j}^{(k'+1)}-\boldsymbol{\lambda }_{i,j}^{(k')} \right\| ^2}} \nonumber
		\\
		&\overset{\left( b \right)}{\leqslant} {\bar{L}\bar{k}\left\| \boldsymbol{\lambda }^{\left( k \right)}-\boldsymbol{\lambda }^{\left( k-1 \right)} \right\| ^2+\bar{L}\sum_{k'=k-\bar{k}}^{k'=k-1}{\left\| \boldsymbol{\lambda }^{(k'+1)}-\boldsymbol{\lambda }^{(k')} \right\| ^2}} \nonumber
	\end{align}
	where (a) is because $\boldsymbol{\lambda}_{i,j}$ appears in both $\boldsymbol{\lambda}_{i}$ and $\boldsymbol{\lambda}_{j}$. (b) is due to $\frac{\left( L_i+L_j \right)}{2}\leqslant \bar{L}:= \max_i L_i$.
	
	Combining all the bound above, and summing the inequality above from iteration $0$ to $k$, we get:
	\begin{align}
		&-D( \boldsymbol{\lambda }^{( k+1 )} ) \leqslant -D( \boldsymbol{\lambda }^{( 0 )} ) -( \frac{1}{\rho}-\frac{L}{2} ) \sum_k{\left\| \boldsymbol{\lambda }_{}^{(k-1)}-\boldsymbol{\lambda }_{}^{(k)} \right\| ^2} \nonumber
		\\
		&+\sum_k{\sum_{k'=k-\bar{k}}^{k'=k-1}{\bar{L}\left\| \boldsymbol{\lambda }_{}^{(k'+1)}-\boldsymbol{\lambda }_{}^{(k')} \right\| ^2}}+\bar{L}\bar{k}\sum_k{\left\| \boldsymbol{\lambda }^{( k+1 )}-\boldsymbol{\lambda }^{( k )} \right\| ^2}
		\\
		&\leqslant D( \boldsymbol{\lambda }^{( 0 )} ) -( \frac{1}{\rho}-\frac{L}{2} ) \sum_k{\left\| \boldsymbol{\lambda }_{}^{(k-1)}-\boldsymbol{\lambda }_{}^{(k)} \right\| ^2}
		\\
		&+\bar{L}\bar{k}\sum_k{\left\| \boldsymbol{\lambda }^{( k'+1 )}-\boldsymbol{\lambda }^{( k' )} \right\| ^2}+\bar{L}\bar{k}\sum_k{\left\| \boldsymbol{\lambda }^{( k+1 )}-\boldsymbol{\lambda }^{( k )} \right\| ^2}	\nonumber
	\end{align}
	The proof is then completed.
\end{proof}


\bibliographystyle{elsarticle-num}
\bibliography{example}






\end{document}